\documentclass[twocolumn,amsthm]{autart}
\usepackage{natbib}
\usepackage[T1]{fontenc}
\usepackage[latin1]{inputenc}
\usepackage[english]{babel}
\usepackage{amssymb,amsfonts}
\usepackage{graphicx}
	\graphicspath{{./}}
 	\DeclareGraphicsExtensions{.pdf}
\usepackage{subcaption}
\usepackage{newtxtext,newtxmath}
\usepackage[noend]{algpseudocode}
	\algrenewcommand\algorithmicrequire{\textbf{Input:}}
	\algrenewcommand\algorithmicensure{\textbf{Output:}}
\usepackage[colorlinks=true,linkcolor=black,citecolor=blue,urlcolor=blue]{hyperref}
\usepackage{booktabs}
	\newcommand{\ra}[1]{\renewcommand{\arraystretch}{#1}}

\newcommand{\eps}{\varepsilon}
\DeclareMathOperator{\A}{\mathcal A}

\allowdisplaybreaks

\edef\endfrontmatter{%
 \unexpanded\expandafter{\endfrontmatter}
 \noexpand\endNoHyper 
}

\journal{ArXiv, accepted in Automatica.}

\begin{document}
\begin{frontmatter}

\title{On Algorithms verifying Initial-and-Final-State Opacity:\\ Complexity, Special Cases, and Comparison}

\thanks{Corresponding author: T. Masopust.}

\author{Tom{\' a}{\v s}~Masopust}\ead{tomas.masopust{@}upol.cz} \and
\author{Petr Osi\v{c}ka}\ead{petr.osicka@upol.cz}

\address{Faculty of Science, Palacky University Olomouc, Czechia}

\begin{abstract}
	Opacity is a general framework modeling security properties of systems interacting with a passive attacker. Initial-and-final-state opacity (IFO) generalizes the classical notions of opacity, such as current-state opacity and initial-state opacity. In IFO, the secret is whether the system evolved from a given initial state to a given final state or not. There are two algorithms for IFO verification. One arises from a trellis-based state estimator, which builds a semigroup of binary relations generated by the events of the automaton, and the other is based on the reduction to language inclusion. The time complexity of both algorithms is bounded by a super-exponential function, and it is a challenging open problem to find a faster algorithm or to show that no faster algorithm exists. We discuss the lower-bound time complexity for both general and special cases, and use extensive benchmarks to compare the existing algorithms.
\end{abstract}


\end{frontmatter}
	
\section{Introduction}
	Opacity is a framework to specify security properties of systems interacting with a passive attacker. The attacker is an observer with a complete knowledge of the structure of the system and with a limited observation of its behavior. The secret is given either as a set of states or as a set of strings. The attacker estimates the behavior of the system based on its observations, and the system is opaque if the attacker never ascertains the secret. For more information, we refer to \cite{JacobLF16}.

	For automata models, several notions of opacity have been investigated in the literature, including current-state opacity (CSO), initial-state opacity (ISO), and initial-and-final-state opacity (IFO). Whereas CSO prevents the attacker from discovering that the system is currently in a secret state, ISO prevents the attacker from discovering that the system started in a secret state. \cite{WuLafortune2013} introduced IFO as a generalization of both CSO and ISO.

 	The applicability of theoretical results depends upon the computational efficiency of the developed algorithms. The fundamental question concerns the time complexity of the algorithms. In particular, a tight complexity characterization provides an insight into the size of instances the algorithms are able to handle in an acceptable time and helps us understand the notion of an optimal algorithm. For (super\nobreakdash-)exponential algorithms, the question of the existence of faster algorithms is particularly important, and lies in the core of the P vs. NP problem.

	For selected notions of opacity, \cite{Wintenberg2021} and \cite{BalunMO2024}  show that the existing algorithms can handle models with a few hundreds or thousands of states, which is far from industrial needs. For IFO, we show that the existing algorithms can handle only very small instances. In Figure~\ref{nfa5states}, we give an example of an automaton with five states and 13 events that our tools implementing current algorithms are not able to solve in 48 hours. 

	One of the purposes of this paper is to show that the situation is not hopeless as the example suggests. By making use of efficient tools for language-inclusion testing, we are able to verify IFO for significantly larger instances. In particular, we are able to verify the example of Figure~\ref{nfa5states} in a few milliseconds. To the best of our knowledge, no research on efficiency and optimization of algorithms for opacity verification has been done in the literature so far.

 	Whereas the worst-case time complexity to verify most of the opacity notions is exponential, and tight under the assumptions discussed below, the time complexity of current algorithms to verify IFO is super-exponential and it is not known whether there is an exponential-time algorithm or whether the super-exponential time complexity is tight. We show that for the algorithms discussed so far in the literature, the super-exponential time complexity is tight. 

 	The (non)existence of an exponential-time algorithm remains a challenging open problem. If there were such an algorithm, our results imply that it would require new ideas and verification techniques. To prove the nonexistence of such an algorithm is even more intricate, because no techniques are known regarding how to show that there is no algorithm of a given complexity; consider, for instance, the questions of the separation of complexity classes in complexity theory or of the existence of a polynomial algorithm for prime factorization.
 	
 	{\bf Current algorithms to verify IFO.}
 	There are two algorithms verifying IFO in the literature---the trellis-based algorithm of \cite{WuLafortune2013}, and the algorithm using a reduction to language inclusion of \cite{BalunMO23}. 
 
 	The trellis-based algorithm uses a state estimator of \cite{SabooriH2013} building a semigroup of binary relations defined by events of a given automaton. 
 	The inclusion algorithm is based on the language inclusion of two automata. It computes the product of one automaton with the complement of the other, and checks emptiness. For nondeterministic automata, the involved complementation requires the construction of the observer.
 
 	Considering time complexity, the upper bounds on existing algorithms are super-exponential of order $O^*(2^{n^2})$, where $n$ is the number of states of the automaton. The IFO verification problem is {\sc PSpace}-complete, and the generally accepted assumption that {\sc PTime} is different from {\sc PSpace} implies that there is no polynomial-time algorithm to verify IFO.
 	
 	However, there could be a sub-exponential-time algorithm. To (conditionally) exclude its existence, \cite{ImpagliazzoP01} formulated a strong exponential time hypothesis motivated by the fact that there is so far no algorithm solving SAT significantly faster than trying all possible truth assignments. The hypothesis states that, for any constant $c < 2$, there is a sufficiently large $k$ such that $k$-SAT (SAT with each clause of the formula containing no more than $k$ literals) cannot be solved in time $O(c^n)$, where $n$ is the number of variables of the formula.
 	Under this hypothesis, \cite{BalunMO23} showed that there is no algorithm to verify IFO of an $n$\nobreakdash-state automaton in time $O^*(2^{n/(2+\eps)})$, for any $\eps > 0$.
 	
 	The question that remains is whether we can verify IFO in exponential time $O^*(2^n)$ rather then in super-exponential time $O^*(2^{n^2})$. This is a significant difference; compare, for example, $2^5=32$ and $2^{5^2}=33,554,432$.
 
	{\bf Our contributions.} 
	First, we discuss the time complexity of both existing IFO-verification algorithms and show that their tight time complexity is super-exponential. Then, we show that the trellis-based algorithm is a special case of the other. 

	Second, we discuss several special cases. Namely, we show that (i) for the set of nonsecret pairs in the form of a Cartesian product, the complexity drops to $O^*(2^{n})$, (ii) for deterministic automata, the complexity drops to $O^*(2^{n \log n})$, (iii) for automata satisfying the observer property of \cite{WW96}, the complexity is polynomial, and (iv) for automata where all cycles are in the form of self-loops, the complexity drops to $O^*(2^{n(n+1)/2})$, resp. to $O^*((n+1)!)$ for deterministic automata, see Table~\ref{table02}.
	
	\begin{table}
		\centering
		\ra{1.1}
		\begin{tabular}{ccl}
			\toprule
			Lower bound & Upper bound & Condition\\
			\midrule
			$\Omega(2^{n^2})$ & $O^*(2^{n^2})$ & none \\
			$O^*(2^{n})$ & $O^*(2^{n})$ & $Q_{NS}=I_{NS}\times F_{NS}$\\
			$O^*(2^{n\log n})$ &  $O^*(2^{n\log n})$ & deterministic\\			
			poly & poly & observer property\\			
			$O^*(2^{n(n+1)/2})$ & $O^*(2^{n(n+1)/2})$ & self-loops nondet.\\			
			$O^*((n+1)!)$ & $O^*((n+1)!)$ & self-loops det.\\			
			\bottomrule
		\end{tabular}
		\smallskip
		\caption{Upper and lower bounds on the time complexity of existing algorithms for IFO verification. If they coincide, then the bound is tight. Here $n$ stands for the number of states of the automaton.}
		\label{table02}
	\end{table}

 	Third, we design new algorithms using advanced tools for language-inclusion testing and create extensive benchmarks, based on real data, to compare the existing and our new IFO-verification algorithms. Our results show that the new algorithms perform better and are able to verify larger instances. The algorithms and benchmarks are available at \url{https://apollo.inf.upol.cz:81/masopust/ifo-benchmarks}.

\section{Preliminaries}\label{prelim}
 	We assume that the reader is familiar with automata theory, see \cite{HopcroftU79}. For a set $S$, the cardinality of $S$ is denoted by $|S|$ and the power set of $S$ by $2^{S}$. An alphabet $\Sigma$ is a finite nonempty set of events, partitioned into $\Sigma_o$ and $\Sigma_{uo}$, of the observable and unobservable events, respectively. 
 	A string over $\Sigma$ is a finite sequence of events from $\Sigma$. The set of all strings over $\Sigma$ is denoted by $\Sigma^*$, and $\varepsilon$ denotes the empty string. A language $L$ over $\Sigma$ is a subset of $\Sigma^*$. 
	
	{\bf Automata.}
	An {\em automaton\/} over an alphabet $\Sigma$ is a triple $\A = (Q,\Sigma,\delta)$, where $Q$ is a finite set of states and $\delta \colon Q\times\Sigma \to 2^Q$ is a transition function that can be extended to the domain $2^Q\times\Sigma^*$ by induction. The language accepted by $\A$ from a set of states $I\subseteq Q$ by a set of states $F\subseteq Q$ is the set $L_m(\A,I,F) = \{w\in \Sigma^* \mid \delta(I,w)\cap F \neq\emptyset\}$ and the language generated by $\A$ from $I$ is the set $L(\A,I) = L_m(\A,I,Q)$. The automaotn $\A$ is {\em deterministic\/} if $|\delta(q,a)|\le 1$ for every state $q\in Q$ and every event $a \in \Sigma$.

	For two sets $I$ and $F$ of states, we use the notation $\A[I,F]$ to denote a copy of $\A$ where $I$ is the set of initial states and $F$ is the set of final states. Notice that $\A[I,F]$ is the classical {\em nondeterministic finite automaton\/} (NFA), and that it is a {\em deterministic finite automaton\/} (DFA) if it is deterministic and has a single initial state. To specify the components of $\A[I,F]$, we use the standard notation $\A=(Q,\Sigma,\delta,I,F)$. If a set is a singleton, we simply write its element; for instance, if $I=\{i\}$, we write $\A=(Q,\Sigma,\delta,i,F)$.
	
	For automata $\A_i = (Q_i,\Sigma_i,\delta_i)$, for $i=1,\ldots,n$ and $n\ge 2$, if the sets $Q_i$ and $Q_j$ are disjoint whenever $i\neq j$, then the \emph{nondeterministic union\/} of $\A_1, \ldots, \A_n$ is the automaton $\A = (\bigcup_{i=1}^{n} Q_i, \bigcup_{i=1}^{n} \Sigma_i, \bigcup_{i=1}^{n} \delta_i)$ formed by union of components; since every function is a relation by definition, we can view the transition function $\delta_i$ as a subset of $Q_i \times \Sigma_i \times Q_i$, which justifies the definition of $\A$. Then, for every subsets $I$ and $F$ of $\bigcup_{i=1}^{n} Q_i$, we have $L(\A,I) = \bigcup_{i=1}^{n} L(\A_i,I\cap Q_i)$ and $L_m(\A,I,F) = \bigcup_{i=1}^{n} L_m(\A_i, I\cap Q_i, F\cap Q_i)$.
	
	The \emph{disjoint union\/} of automata $\A_1, \ldots, \A_n$ first makes the state sets of the automata pairwise disjoint by a suitable renaming of the states, together with the corresponding adjustment of the transition functions, and then performs the nondeterministic union on the resulting automata.
	
	{\bf Projections of strings, languages, and automata.}
	A {\em projection\/} $P\colon \Sigma^* \to \Sigma_o^*$ is a morphism for concatenation defined by $P(a) = \eps$, for $a\in \Sigma\setminus\Sigma_{o}$, and $P(a)=a$, for $a\in \Sigma_o$. Intuitively, the action of $P$ is to erase all unobservable events. We lift projections from strings to languages in the usual way.

 	Let $\A$ be an automaton over $\Sigma$, and let $P\colon \Sigma^* \to \Sigma_o^*$ be a projection. The {\em projected automaton\/} of $\A$, denoted by $P(\A)$, is obtained from $\A$ by replacing every transition $(q,a,r)$ with $(q,P(a),r)$, and by using the classical elimination of $\eps$\nobreakdash-transitions. Then $P(\A)$ is an automaton over $\Sigma_o$, preserving the same observable behavior as $\A$, and with the same set of states as $\A$. The automaton $P(\A)$ can be constructed from $\A$ in polynomial time \citep{HopcroftU79}. 
 	
 	For an NFA $\A$, the reachable part of the DFA constructed from $P(\A)$ by the standard subset construction is called the {\em observer} of $\A$. In the worst case, the observer of $\A$ has exponentially many states compared with $\A$; see \cite{JiraskovaM12}.

	{\bf IFO.}
 	An automaton $\A=(Q,\Sigma,\delta)$ is \emph{initial-and-final-state opaque} (IFO) with respect to sets $Q_S,Q_{NS}\subseteq Q \times Q$ of secret and nonsecret pairs of states, respectively, and a projection $P\colon \Sigma^*\rightarrow \Sigma_o^*$, if for every secret pair $(s,t) \in Q_S$ and every string $w \in L_m(\A,s,t)$, there is a nonsecret pair $(s',t') \in Q_{NS}$ and a $w' \in L_m(\A,s',t')$ such that $P(w) = P(w')$.

	{\bf Asymptotic complexity.}
 	Let $g\colon\mathbb{R}\to\mathbb{R}$ be a real function. The class $O(g(n)) = \{f\colon\mathbb{R}\to\mathbb{R} \mid \text{there are } c, n_0 > 0 \text{ such that } 0 \le f(n) \le cg(n), \text{ for every } n \ge n_0 \}$ consists of functions that do not grow asymptotically faster than $g$; intuitively, $O(g(n))$ neglects constant factors. Analogously, $O^*(g(n))=O(g(n) \textrm{poly}(n))$ neglects constant and polynomial factors.
 	The class $o(g(n)) = \{f\colon\mathbb{R}\to\mathbb{R} \mid \text{for every } c > 0 \text{ there is } n_0 > 0 \text{ such that } |f(n)| < cg(n), \text{ for every } n \ge n_0 \}$ consists of functions that grow asymptotically strictly slower than $g$; that is, $f(n) \in o(g(n))$ if and only if $\lim_{n\to\infty} f(n)/g(n) = 0$.
 	The class of functions that do not grow asymptotically slower than $g$ is denoted by $\Omega(g(n))$ and defined by $f(n) \in \Omega(g(n))$ if and only if $g(n) \in O(f(n))$.

 	A function $f(n)$ is \emph{super-exponential\/} if it grows faster than any exponential function of the form $c^n$, where $c$ is a constant; formally, $\lim_{n\to\infty} f(n)/c^n = \infty$ for all constants $c>1$.

	{\bf Semigroup theory and automata.}
	A \emph{semigroup\/} is a set $S$ together with a binary operation $\cdot$ on $S$ that is associative, i.e., $a\cdot(b\cdot c)=(a\cdot b)\cdot c$. 
	
 	For an automaton $\A=(Q,\Sigma,\delta)$ with $n$ states, every event $a\in\Sigma$ defines a \emph{binary relation\/} on $Q$ that can be represented by an \emph{$n\times n$ binary matrix} $(a_{ij})$, where $a_{ij}=1$ if $j\in \delta(i,a)$, and $a_{ij}=0$ otherwise. A \emph{boolean multiplication\/} of binary matrices $(a_{ij})$ and $(b_{ij})$ is defined as the matrix $(c_{ij})$, where $c_{ij} = \max\{a_{ik}b_{kj} \mid k=1,\ldots,n \}$, that is, as the classical matrix multiplication where addition is replaced by maximum. 
	The set of all $n\times n$ binary matrices together with the boolean matrix multiplication forms a semigroup $\mathcal{B}_n$ containing $2^{n^2}$ elements.
	
	Consider the set of binary matrices $\mathcal{G}_{\A} = \{(a_{ij}) \mid a \in \Sigma\}$ corresponding to the events of $\A$, and denote by $\mathcal{B}_{\A}$ the semigroup containing all possible finite products of elements of $\mathcal{G}_{\A}$. Then, $\mathcal{B}_{\A}$ is a subsemigroup of $\mathcal{B}_n$, and $\mathcal{G}_{\A}$ is a set of generators of $\mathcal{B}_{\A}$. In particular, every string $w\in \Sigma^*$ defines a binary matrix $(w_{ij})\in \mathcal{B}_{\A}$ representing the relation on $Q$ such that $w_{ij}=1$ if and only if $j\in\delta(i,w)$. 
	
	A fundamental question in semigroup theory is the minimum number of generators of a semigroup. Despite an intensive study of this question for the semigroup $\mathcal{B}_n$ for more than 60 years, the answer is known only for $n\le 8$, see \cite{hivert2021minimal} or the sequence A346686 of the On-Line Encyclopedia of Integer Sequences (OEIS). Although the minimum number of generators of $\mathcal{B}_n$ is unknown for $n\ge 9$, \cite{Devadze1968} claimed without proof, and \cite{Konieczny2011} proved, that it grows super-exponentially with respect to $n$. A lower bound on the number of generators of $\mathcal{B}_n$ is $2^{{n^2}/{4} - O(n)}/(n!)^2$, see \cite[Corollary~3.1.8]{hivert2021minimal}. 

	Proposition~6 and Theorem~7 in \cite{KimRoush1978} further show that (i) for every odd natural number $n$, there are two binary matrices generating a subsemigroup of $\mathcal{B}_n$ with $((n^2-1)/4) \cdot 2^{(n^2-1)/4}$ elements, and (ii) for every integer-valued function $f$ such that $f(n)>n$ and $f(n)/n^2$ tends to zero, the average size of the semigroup generated by two randomly chosen $n\times n$ binary matrices, each with $f(n)$ ones, is at least $2^{(n^2/4)+o(n^2)}$.

	{\bf Semigroup theory and deterministic automata.}
	Similarly as strings over an automaton correspond to binary relations on its states, the strings over a deterministic automaton $\A=(Q,\Sigma,\delta)$ correspond to partial functions on states. Namely, every $w\in \Sigma^*$ defines a partial function $\delta_w$ on $Q$ as follows: for every $q\in Q$, $\delta_w(q)=\delta(q,w)$. 
	Partial functions on $Q$ are called \emph{partial transformations\/} and, together with the composition of functions, form a so-called {\em partial-transformation semigroup\/} denoted by $\mathcal{PT}_n$. 
	The semigroup $\mathcal{PT}_n$ has $(n+1)^n$ elements and is generated by four transformations, see \cite[Exercises 12--13, page~41]{howie1995}. 
	
	The reader can find more information on the size of the subsemigroups of $\mathcal{PT}_n$ generated by less than four transformations in \cite{HolzerK04} and \cite{KrawetzLS05}. For total transformations, we refer to \cite{Salomaa60} or \cite{denes1966}.

\section{Trellis-Based IFO Verification}
\label{secSemigroupVer}
 	In this section, we reformulate the trellis-based algorithm of \cite{WuLafortune2013} in terms of semigroups of binary relations, see Algorithm~\ref{alg-trellis}. 

 	Given an automaton $\A$ over $\Sigma$, sets $Q_S$ and $Q_{NS}$ of secret and nonsecret pairs of states, and a set of observable events $\Sigma_o\subseteq \Sigma$, the algorithm checks whether every matrix $(w_{ij}) \in \mathcal{B}_{\A}$ generated by the events of $\A$ satisfies the condition: {\it If there is 1 at a position $w_{sf}$ corresponding to a secret pair $(s,f)\in Q_S$, then there is 1 at a position $w_{s'f'}$ for some nonsecret pair $(s',f')\in Q_{NS}$}. 
 
	 \begin{algorithm}\hrule height .08em\vspace{-5pt}
		\caption{Semigroup-based IFO verification}
		\label{alg-trellis}
		\begin{algorithmic}[1]
			\vspace{2pt}\hrule\vspace{3pt}
			\Require An automaton $\A=(Q,\Sigma,\delta)$,
	    	\Statex sets $Q_S,Q_{NS}\subseteq Q\times Q$ of secret and nonsecret pairs,
	    	\Statex and an alphabet $\Sigma_o\subseteq \Sigma$ of observable events.
		
			\Ensure 
		    {\tt true} iff $\A$ is IFO wrt $Q_S$, $Q_{NS}$, and $P\colon \Sigma^*\to\Sigma_o^*$.
		    
			\State Construct the projected automaton $P(\A)$

			\State Compute the elements of the semigroup $\mathcal{B}_{P(\A)}$, generated by the events of $\Sigma_o$, one by one
		   
			\For {every newly generated element $(w_{ij}) \in \mathcal{B}_{P(\A)}$}
		    	\For {every secret pair $(s,f)\in Q_S$}
		     		\If {$w_{sf}=1$ and there is no nonsecret pair 
		      			\State $(s',f')\in Q_{NS}$ such that $w_{s'f'}=1$}
		      			\State \Return {\tt false}
		     		\EndIf
		    	\EndFor
		   \EndFor
		   \State \Return {\tt true}
		  \end{algorithmic}
		  \hrule height .08em
	 \end{algorithm}
 
 	Considering the worst-case time complexity of Algorithm~\ref{alg-trellis}, the construction of the projected automaton $P(\A)$ is polynomial in the number of states of $\A$, whereas the computation of the semigroup $\mathcal{B}_{P(\A)}$ depends on the structure of $\A$ and may vary significantly. Taking the size of $\mathcal{B}_{P(\A)}$ as a parameter, the time complexity of Algorithm~\ref{alg-trellis} is dominated by the cycle on lines~3--7, giving the overall time complexity $O(|\mathcal{B}_{P(\A)}|\cdot n^2)$, where $n$ is the number of states of $\A$; indeed, the inner loop on lines~4--7 may be verified by a single scan of all elements of the matrix $(w_{ij})$. 
 
 	Before we discuss the maximal size of the semigroup $\mathcal{B}_{P(\A)}$, notice that for instances that are IFO, the algorithm has to build the whole semigroup $\mathcal{B}_{P(\A)}$, whereas for instances that are not IFO, the algorithm terminates as soon as a matrix that fails the condition is generated.
 
 	We now discuss the lower-bound complexity of Algorithm~\ref{alg-trellis}.
  
	\begin{thm}\label{thm5}
		For every $n\ge 1$, there is an automaton $\A_n$ with $n$ states and $2^{n^2}$ events, all of which are observable, for which Algorithm~\ref{alg-trellis} needs at least $2^{n^2}$ steps.
	\end{thm}  
	\begin{proof}
		For every $n\ge 1$, the automaton $\A_n=(Q_n,\Sigma_n,\delta_n)$ is constructed from the semigroup $\mathcal{B}_n$ by taking, for every matrix $(a_{ij})$ of $\mathcal{B}_n$, an event $a$ that connects state $i$ to state $j$ if and only if $a_{ij}=1$. Formally, $Q_n=\{1,2,\ldots,n\}$, $\Sigma_n = \{a \mid (a_{ij})\in \mathcal{B}_n\}$, and for $i,j\in Q_n$ and $a\in \Sigma_n$, we define $j \in \delta_n(i,a)$ if and only if $a_{ij}=1$. Then, the semigroup $\mathcal{B}_{\A_n}$ coincides with the semigroup $\mathcal{B}_n$. Consequently, if the IFO instance for the automaton $\A_n$ is positive, then Algorithm~\ref{alg-trellis} has to verify all matrices of $\mathcal{B}_{n}$, of which there are $2^{n^2}$.
	\end{proof}
		
	From the construction, we may observe that Algorithm~\ref{alg-trellis} has to make at least  $2^{n^2}$ steps even if we define an event for every generator of $\mathcal{B}_n$ rather than for every element of  $\mathcal{B}_n$.
	
	For an illustration, consider the semigroup $\mathcal{B}_2$ with its three generators 
	$(a_{ij})=\big(\begin{smallmatrix}
		0 & 1 \\
		1 & 0
	\end{smallmatrix}\big)$,
	$(b_{ij})=\big(\begin{smallmatrix}
		1 & 0 \\
		1 & 1
	\end{smallmatrix}\big)$, and
	$(c_{ij})=\big(\begin{smallmatrix}
		1 & 0 \\
		0 & 0
	\end{smallmatrix}\big)$.
	Then, the automaton $\A_2$ constructed in Theorem~\ref{thm5} contains two states, $Q_2=\{1,2\}$, three events, $\Sigma_2=\{a,b,c\}$, and the transition function $\delta_2$ is obtained from the matrices as follows: the matrix $(a_{ij})$ results in the transitions $\delta(1,a)=\{2\}$ and $\delta(2,a)=\{1\}$, the matrix $(b_{ij})$ results in the transitions $\delta(1,b)=\{1\}$ and $\delta(2,b)=\{1,2\}$, and the matrix $(c_{ij})$ results in the transition $\delta(1,c)=\{1\}$. 
	Since the matrices $(a_{ij})$, $(b_{ij})$, $(c_{ij})$ are generators of the semigroup $\mathcal{B}_2$, each of the 16 elements of $\mathcal{B}_2$ can be written as a multiplications of (some of) the generators. For example, the matrix $\big(\begin{smallmatrix}
		1 & 1 \\
		1 & 1
	\end{smallmatrix}\big)=(a_{ij})(b_{ij})(a_{ij})(b_{ij})$. This matrix corresponds to the relation on states of $\A_2$ defined by the string $abab$; indeed, we have $\delta(i,abab)=\{1,2\}$, for $i\in\{1,2\}$.
	Consequently, $\mathcal{B}_{\A_2}=\mathcal{B}_2$, and if $\A_2$ satisfies IFO, then Algorithm~\ref{alg-trellis} generates all elements of $\mathcal{B}_2$ and requires thus at least 16 steps to verify each of the 16 elements of $\mathcal{B}_2$.
	
	Analogously, considering $\mathcal{B}_5$ with its 13 generators results in the automaton $\A_5$ with five states and 13 events depicted in Figure~\ref{nfa5states}. This automaton is of particular interest, because a tool implementing Algorithm~\ref{alg-trellis} with the automaton $\A_5$, sets $Q_S=Q_{NS}=\{(1,1),(2,2),(3,3),(4,4),(5,5)\}$, and the alphabet $\Sigma_o=\Sigma$ as input does not terminate in 48 hours.\footnote{The choice of $Q_S=Q_{NS}$ only ensures that $\A_5$ satisfies IFO. The same behavior would be observed for any other choice of the sets $Q_S$ and $Q_{NS}$ for which the automaton $\A_5$ satisfies IFO.}
	
	\begin{figure}
		\centering
		\includegraphics[scale=.45]{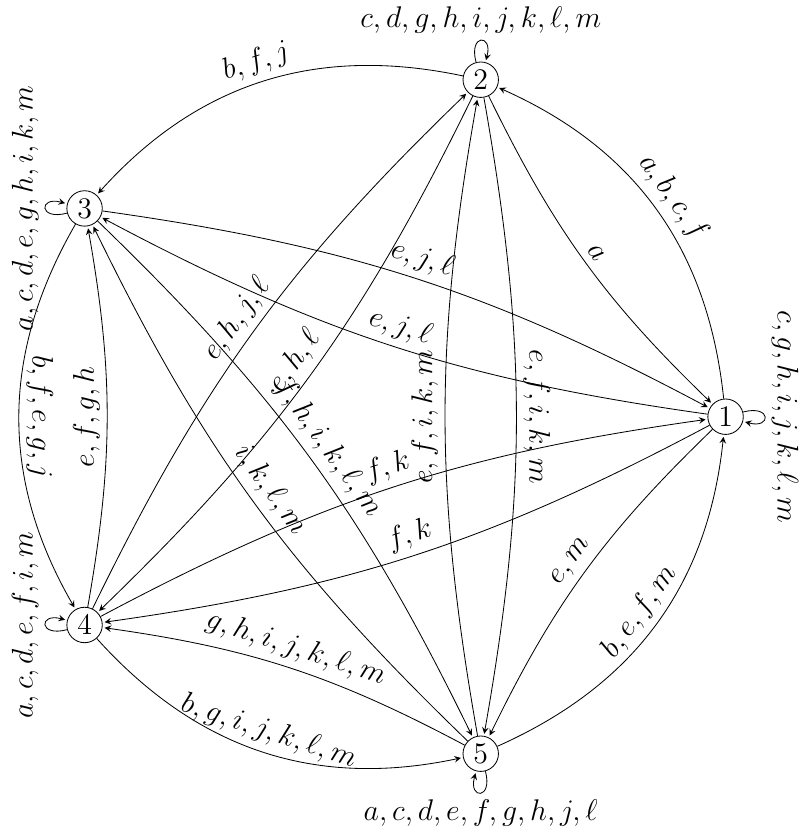}
		\caption{The automaton $\A_5$ on which Algorithm~\ref{alg-trellis} makes at least $2^{25}$ steps and fails to terminate in 48 hours.}
		\label{nfa5states}
	\end{figure}

	Recall that every semigroup $\mathcal{B}_n$ has a minimum number of generators, which are known for $n\le 8$. For larger $n$, neither the generators nor their minimum number is known. However, the number is at least $2^{{n^2}/{4} - O(n)}/(n!)^2$, which improves Theorem~\ref{thm5} by decreasing the number of events of the automaton $\A_n$ from $2^{n^2}$ to an unknown number lower-bounded by the super-exponential function $2^{{n^2}/{4} - O(n)}/(n!)^2$.
 
	The reader may wonder whether the complexity of Algorithm~\ref{alg-trellis} drops if the number of events grows asymptotically slower in the number of states than super-exponentially. Unfortunately, Proposition~6 of \cite{KimRoush1978} implies that the complexity remains super-exponential even if we consider only automata with two observable events.
 
	\begin{cor}
		For every odd $n\ge 1$, there is an automaton $\A_n$ with $n$ states and two observable events, for which Algorithm~\ref{alg-trellis} runs in time $\Omega(n^2 2^{(n^2-1)/{4}})$. \qed
	\end{cor}  
 
 	The results so far consider the worst-case time complexity. From the practical point of view, it is of interest to consider the average time complexity of Algorithm~\ref{alg-trellis}. As an immediate consequence of Theorem~7 of \cite{KimRoush1978}, we obtain the following result. 
 
	\begin{thm}\label{thm4}
		Let $f$ be an integer-valued function satisfying $f(n)>n$ and $\lim_{n\to\infty} f(n)/n^2 = 0$. Consider the set $\mathbb{A}_n$ of all automata with the state set $\{1,\ldots,n\}$ and two observable events where either event appears on $f(n)$ transitions. Then, the average runtime of Algorithm~\ref{alg-trellis} on an automaton chosen uniformly at random from $\mathbb{A}_n$ is $\Omega(2^{n^2/4 +o(n^2)})$. \qed
	\end{thm}
 
 	We further discuss experimental results on an extensive number of data in Section~\ref{secExp}. (Results of experiments on randomly generated data are available in the git repository.)

\section{Inclusion-Based IFO Verification}
\label{secObserverVer}
	The algorithm of \cite{BalunMO23} reduces the IFO verification to language inclusion: given two NFAs $\A_1$ and $\A_2$, is $L_m(\A_1)\subseteq L_m(\A_2)$? Although there are different optimization techniques discussed in the literature, all existing algorithms for language inclusion are based on the classical observer (aka subset or powerset) construction. 
 
 	For an automaton $\A=(Q,\Sigma,\delta)$ and a relation $S\subseteq Q\times Q$, we denote by $S'$ the relation over $Q\times 2^Q$ defined by
	\begin{align}\label{set01}
		S' & = \{ (s,T) \mid f\in T \text{ if and only if } (s,f) \in S \}
	\end{align}
	that clusters the elements of $S$ with respect to the domain.
 	To reduce IFO to language inclusion, let $Q_{S},Q_{NS}\subseteq Q\times Q$ be the sets of secret and nonsecret pairs, respectively. These sets can be clustered as shown in \eqref{set01}, resulting in the sets $Q_S'$ and $Q_{NS}'$. Based on these sets, we define the languages
 	\begin{equation}\label{eq3}
		\begin{aligned}
			L_S & := \bigcup_{(s,T)\in Q_S'} L_m(\A,s,T)  \quad\text{ and } \\
			L_{NS} & := \bigcup_{(s,T)\in Q_{NS}'} L_m(\A,s,T)\,.
		\end{aligned}
	\end{equation}
	The verification of IFO now reduces to the verification of the inclusion $P(L_S)\subseteq P(L_{NS})$, see \cite{BalunMO23}.
 
	In the sequel, we distinguish two cases of the verification of language inclusion: (i) the classical approach based on the observer construction, and (ii) approaches reducing the state space of the constructed observer.

	{\bf Observer-based IFO verification.}
	The languages $L_S$ and $L_{NS}$ of \eqref{eq3} can be represented by NFAs $\A_S$ and $\A_{NS}$, respectively, consisting of no more than $n$ copies of the automaton $\A$, with possibly different initial and final states; see Algorithm~\ref{alg1}. As a result, the NFAs $\A_S$ and $\A_{NS}$ have no more than $n^2$ states each, and the classical verification of the language inclusion $P(L_S)\subseteq P(L_{NS})$ based on the observer construction is of time complexity $O(n^2 2^{n^2}) = O^*(2^{n^2})$. 
	
	Note that the upper bound on the time complexity of Algorithm~\ref{alg1} coincides with the upper bound on the time complexity of Algorithm~\ref{alg-trellis}. It is not a coincidence. We show that Algorithm~\ref{alg-trellis} is a special case of Algorithm~\ref{alg1}.

	\begin{algorithm}\hrule height .08em\vspace{-5pt}
		\caption{Classical observer-based IFO verification}
		\label{alg1}
		\begin{algorithmic}[1]
			\vspace{2pt}\hrule\vspace{3pt}
			\Require An automaton $\A=(Q,\Sigma,\delta)$, 
			\Statex sets $Q_S,Q_{NS}\subseteq Q\times 2^Q$ of secret and nonsecret pairs, 
			\Statex and the alphabet $\Sigma_o\subseteq \Sigma$ of observable events.

			\Ensure {\tt true} iff $\A$ is IFO wrt $Q_S$, $Q_{NS}$, and $P\colon \Sigma^*\to\Sigma_o^*$.

			\State Let $\A_S$ be disjoint union of $\A[s,T]$, for $(s,T)\in Q_S$. 
   
			\State Let $\A_{NS}$ be disjoint union of $\A[s,T]$, for $(s,T)\in Q_{NS}$.
			
			\State Compute the observer $\A_{NS}^{obs}$ of $\A_{NS}$
			
			\State \Return $L_m(\A_S) \times L_m(\A_{NS}^{obs}) = \emptyset$?
		\end{algorithmic}
		\hrule height .08em
	\end{algorithm}
 
	\begin{lem}\label{lem9}
  		For every automaton $\A=(Q,\Sigma,\delta)$, an alphabet $\Sigma_o\subseteq \Sigma$ of observable events, and sets $Q_S,Q_{NS}\subseteq Q\times 2^Q$ of secret and nonsecret pairs of states, respectively, there is an automaton $\A'$ and a bijection between the elements of the semigroup $\mathcal{B}_{P(\A)}$ and the states of the observer of $\A'$.
 	\end{lem}
 	\begin{proof}
		Let $P\colon \Sigma^* \to \Sigma_o^*$ be the projection erasing unobservable events, and let $\mathcal{G}_{\A} = \{ (a_{ij}) \mid a \in \Sigma_o$, and $a_{ij}=1$ if and only if $j\in \delta(i,a)\}$ be the set of binary matrices corresponding to the events of $\Sigma_o$. Then, $\mathcal{G}_{\A}$ is a set of generators of $\mathcal{B}_{P(\A)}$, and hence every matrix $(m_{ij}) \in \mathcal{B}_{P(\A)}$ is a product of some of the matrices of $\mathcal{G}_{\A}$; say $(m_{ij})=({a_1})({a_2})\cdots ({a_k})$. Let $m=a_1a_2\cdots a_k$ be the corresponding string over $\Sigma_o$.
 		We denote the $i$th row of the matrix $(m_{ij})$ by $(m_{i*})$. For $i\in Q$, let $\A_i$ denote a copy of $\A$ with $i$ being the single initial state.  Then, the positions of ones in $(m_{i*})$ correspond to the states of $\A_i$ reachable from the state $i$ under the strings from $P^{-1}(m)$, which corresponds to the state of the observer of $\A_i$ reachable from the state $\{i\}$ under the string $m$.
 		
 		We define the automaton $\A'$ as a disjoint union of automata $\A_i$, for all $i\in Q$, and denote by $I'$ the initial states of $\A'$ formed by the initial states of individual $\A_i$'s as renamed by the operation of disjoint union. Then, the positions of ones in $(m_{ij})$ correspond to the state of the observer of $\A'$ reachable from $I'$ under the string $m$. 
 		
 		On the other hand, the state of the observer of $\A'$ reachable from the initial state $I'$ under a string $m$ corresponds to the ones in the matrix $(m_{ij})$. 
 		
 		Thus, there is a one-to-one correspondence between the elements of $\mathcal{B}_{P(\A)}$ and the states of the observer of $\A'$.
	\end{proof}
 
	As a consequence of Lemma~\ref{lem9}, we obtain for Algorithm~\ref{alg1} the same results as for Algorithm~\ref{alg-trellis}. Namely, for every $n\ge 1$, there is an automaton $\A_n$ with $n$ states and super-exponentially many events with respect to $n$ such that the observer of the NFA $\A_{NS}$ of Algorithm~\ref{alg1} has $2^{n^2}$ states; hence, Algorithm~\ref{alg1} has to make at least $2^{n^2}$ steps.
	Further, the worst-case time complexity of Algorithm~\ref{alg1} is $\Omega(n^2 2^{(n^2-1)/4})$ even if the automata are over a binary alphabet, and its average time complexity is super-exponential in the number of states of the automaton if the number of transitions grows with the number of states.
	Similarly to the tool implementing Algorithm~\ref{alg-trellis}, our tool implementing Algorithm~\ref{alg1} fails to terminate on the automaton of Figure~\ref{nfa5states} in 48 hours.
	
	{\bf Inclusion-based IFO verification.}
 	The reduction to language inclusion is of independent interest, and we formulate it as Algorithm~\ref{alg3}. The inclusion problem has been widely investigated in the literature, resulting in many techniques and tools. Algorithm~\ref{alg3} is thus a class of algorithms using different tools  on line~3. We consider the state-of-the-art tools in Section~\ref{secExp}. For an illustration, these tools require less than a second to verify whether the five-state automaton of Figure~\ref{nfa5states} is IFO; recall that the tools implementing Algorithms~\ref{alg-trellis} and~\ref{alg1} failed to terminate in 48 hours. 
 	
 	The main idea of the tools is to cut the state space of the constructed observer by keeping only selected states. Imagine, for instance, two sets of states $X\subseteq Y$ computed by the observer. We can keep only $Y$ with the justification that whatever can be computed from $X$ can also be computed from $Y$.
 	
	\begin{algorithm}\hrule height .08em\vspace{-5pt}
	\caption{General inclusion-based IFO verification}
	\label{alg3}
	\begin{algorithmic}[1]
		\vspace{2pt}\hrule\vspace{3pt}
		\Require An automaton $\A=(Q,\Sigma,\delta)$, sets $Q_S,Q_{NS}\subseteq Q\times 2^Q$, and $\Sigma_o\subseteq \Sigma$.

		\Ensure {\tt true} iff $\A$ is IFO wrt $Q_S$, $Q_{NS}$, and $P\colon \Sigma^*\to\Sigma_o^*$.

		\State Let $\A_S$ be disjoint union of $\A[s,T]$, for $(s,T)\in Q_S$. 

		\State Let $\A_{NS}$ be disjoint union of $\A[s,T]$, for $(s,T)\in Q_{NS}$.

		\State \Return $L_m(\A_S) \subseteq L_m(\A_{NS}^{obs})$?
	\end{algorithmic}
	\hrule height .08em
	\end{algorithm}

	Although the language-inclusion tools are quite efficient in practice, see Section~\ref{secExp}, they are still based on the observer construction. Consequently, the upper bound on the time complexity of Algorithm~\ref{alg3} coincides with that of Algorithms~\ref{alg-trellis} and~\ref{alg1}. However, compared with Algorithms~\ref{alg-trellis} and~\ref{alg1}, it is an open problem whether this time complexity is also tight for Algorithm~\ref{alg3}.

\section{Special Cases}\label{sec6}
	In this section, we discuss the worst-case time complexity of Algorithms~\ref{alg-trellis} and~\ref{alg1} for several special cases.

	{\bf Nonsecret pairs in the form of a Cartesian product.}
 	As an immediate consequence of Algorithm~\ref{alg1}, we obtain that if the set $Q_{NS}$ is a Cartesian product of states, i.e., $Q_{NS} = I \times F$, then the NFA for $L_{NS}$ coincides with the input automaton $\A$ where the states of $I$ are initial and the states of $F$ are final, i.e., with $\A[I,F]$. In particular, the NFA recognizing the language $P(L_{NS})$ has $n$ states, and hence the inclusion $P(L_S)\subseteq P(L_{NS})$ can be tested in time $O(n^2 2^n)=O^*(2^n)$. Thus, we have the following result of \cite{BalunMO23}, which improves a similar result of \cite{WuLafortune2013}.

	\begin{cor}\label{cor4}
		The IFO property of an automaton $\A=(Q,\Sigma,\delta)$ with respect to secret pairs $Q_S \subseteq Q \times Q$, nonsecret pairs $Q_{NS}=I\times F \subseteq Q\times Q$, and a projection $P\colon \Sigma^* \to \Sigma_o^*$ can be verified in time $O(n^2 2^n)$. \qed
	\end{cor}

	{\bf Deterministic automata.}
 	Another special case arises for deterministic automata with full observation. Given a deterministic automaton $\A=(Q,\Sigma,\delta)$ and a set $S\subseteq Q\times 2^Q$, we first study the number of states of the minimal DFA recognizing the language $\bigcup_{(s,T)\in S} L_m(\A,s,T)$.

	\begin{lem}\label{thmUnionDFAupper}
		For a deterministic automaton $\A$ with an $n$-element state set $Q$ and a set $S\subseteq Q \times 2^{Q}$, the minimal DFA recognizing the language $\bigcup_{(s,T)\in S} L_m(\A,s,T)$ has no more than $(n+1)^{n}$ states.
	\end{lem}
	\begin{proof}
		The language $\bigcup_{(s,T)\in S} L_m(\A,s,T)$ can be represented by an NFA $\A'$ constructed as a disjoint union of $\A[s,T]$, for $(s,T)\in S$. Without loss of generality, we assume that $|S|\le |Q|$; if there were two pairs $(s,T_1),(s,T_2)\in S$, we could replace them by a single pair $(s,T_1\cup T_2)$. Thus, $\A'$ consists of $|S|$ deterministic components, and hence the state set of the observer of $\A'$ consists of $|S|$-tuples, where on each position, we have either the state of that component, or the information that the component is no longer active. Therefore, the observer of $\A'$ has at most $(n+1)^{n}$ states. Since the minimal DFA that is language equivalent to $\A'$ does not have more states than the observer of $\A'$, the proof is complete.
	\end{proof}
 
 	Analogously as the semigroup $\mathcal{B}_n$ is related to automata, is the partial-transformation semigroup $\mathcal{PT}_n$ related to deterministic automata without unobservable events. Therefore, replacing the semigroup $\mathcal{B}_n$ in Algorithm~\ref{alg-trellis} by the semigroup $\mathcal{PT}_n$ implies that the tight worst-case time complexity of Algorithm~\ref{alg-trellis} for deterministic automata with full observation is $O((n+1)^n)$. 
 	
 	Similarly, we immediately have the following theorem giving us the worst-case time complexity of Algorithm~\ref{alg1}.
	\begin{thm}
		If $\A$ is an $n$-state deterministic automaton without unobservable events, then the time complexity of Algorithm~\ref{alg1} is bounded from above by $O(n^2 (n+1)^n)$ and from below by $(n+1)^n$.
	\end{thm}
	\begin{proof}
		Consider the languages $L_S$ and $L_{NS}$ as defined in \eqref{eq3}. Lemma~\ref{thmUnionDFAupper} shows that the number of states of the observer of the automaton for the language $L_{NS}$ is at most $(n+1)^n$. Since the number of states of the automaton for $L_S$ is at most $n^2$, the time complexity of Algorithm~\ref{alg1} is $O(n^2 (n+1)^n)$.
	
		For the other part, we consider the semigroup $\mathcal{PT}_n$ generated by four transformations, say $t_1, t_2, t_3, t_4$. We construct a deterministic automaton $\A$ with $n$ states $Q=\{1,\ldots,n\}$ by defining the alphabet $\Sigma=\{t_1, t_2, t_3, t_4\}$ and the transition function $\delta(q,t_i)=t_i(q)$, for every $q\in Q$ and every $t_i \in\{t_1, t_2, t_3, t_4\}$. Since every transformation $t\in \mathcal{PT}_n$ is a composition of the four transformations $t_1, t_2, t_3, t_4$, we can see $t$ as a string $w_t$ over $\Sigma$, and hence $\delta(q,w_t)=t(q)$, for every $q\in Q$. Then, for $Q_{NS}=\{(i,i) \mid i \in Q\}$, the automaton $\A_{NS}$ of Algorithm~\ref{alg1} is a disjoint union of automata $\A_1, \ldots, \A_n$, where $\A_i=\A[i,i]$ is a copy of $\A$ with state $i$ initial and final. Then, the observer of $\A_{NS}$ has as many states as there are transformations in $\mathcal{PT}_n$, which is $(n+1)^n$, and hence Algorithm~\ref{alg1} makes at least that many steps.
	\end{proof}

	{\bf Observer property.}
 	The \emph{observer property\/} was introduced by \cite{WW96} and, as pointed out by \cite{FengW10}, is equivalent to the \emph{observation equivalence\/} in \cite{HennessyM80}. 
	Given a DFA $\A$ over $\Sigma$ generating the language $L$ and accepting the language $L_m$, and a set of observable events $\Sigma_o \subseteq \Sigma$, the projection $P\colon \Sigma^* \to \Sigma_o^*$ is an \emph{$L_m$-observer\/} if for all strings $t\in P (L_m)$ and $s\in L$, whenever the string $P(s)$ is a prefix of $t$, then there exists a string $u\in \Sigma^*$ such that $su\in L_m$ and $P(su) = t$.

 	\cite{wong98} has shown that, under the observer property, the observer of $\A$ does not have more states than the automaton $\A$ itself, and that the observer of $\A$ may be computed in polynomial time with respect to the size of $\A$; see also \cite{WongW04} and \cite{FengW10}. In combination with Algorithm~\ref{alg1}, we have the following.
 
 	\begin{thm}
  		Given a deterministic automaton $\A$ over $\Sigma$, the sets $Q_S$ and $Q_{NS}$ of secret and nonsecret pairs, and a set of observable events $\Sigma_o\subseteq \Sigma$, let $\A_{NS}$ be the disjoint union of automata $\A[s,T]$, for $(s,T)\in Q_{NS}$. If the projection $P\colon \Sigma^* \to \Sigma_o^*$ is an $L_m(\A_{NS})$-observer, then the time to verify IFO is polynomial. \qed
	\end{thm}
	
	\begin{rem}
		Although the definition of the observer property is for DFAs, it can be applied to NFAs. Indeed, every nondeterministic choice $(p,a,q)$ and $(p,a,r)$, with $q\neq r$, can be replaced by three transitions $(p,u,p')$, $(p',a,q)$, and $(p,a,r)$, where $p'$ is a new state and $u$ is a new unobservable event. The construction results in a DFA, the observer of which is isomorphic to the observer of the original automaton.
	\end{rem}

	{\bf Partially ordered automata.}
	Partially ordered automata, aka 1-weak, very weak, linear, acyclic, or extensive automata, see, e.g., \cite{MasopustK21}, are automata where the transition relation induces a partial order on states. This restriction implies that as soon as a state is left during the computation, it is never visited again. Said differently, the only cycles in the automaton are self-loops. Let $\A = (Q, \Sigma, \delta)$ be an automaton. The reachability relation $\le$ on states is defined by setting $p \le q$ if there is a string $w\in \Sigma^*$ such that $q \in \delta(p, w)$. The automaton $\A$ is \emph{partially ordered} if the reachability relation $\le$ is a partial order.
 
	Partially ordered automata recognize a subclass of regular languages strictly included in \emph{star-free languages}, see \cite{BrzozowskiF80} and \cite{mfcs16:mktmmt_full}. Star-free languages are languages definable by \emph{linear temporal logic}, which is a logic widely used as a specification language in automated verification. 

	When we order the states of a partially ordered automaton by the reachability relation $\le$, the matrices of binary relations corresponding to strings over the partially ordered automaton are upper triangular. Upper triangular binary matrices form a semigroup, denoted by $\mathcal{UT}_n$, that has $2^{n(n+1)/2}$ elements and a (unique) minimal generating set with $n(n+1)/2+1$ elements; one of the generators is the identity matrix, see \cite{hivert2021minimal}. Therefore, for every $n\ge 1$, there is a partially ordered automaton $\A_n$ with $n$ states and $n(n+1)/2$ events (we do not need an event corresponding to the identity matrix) such that (i) the semigroup $\mathcal{B}_{P(\A_n)}$ constructed in Algorithm~\ref{alg-trellis} has $2^{n(n+1)/2}$ elements, and (ii) the observer constructed in Algorithm~\ref{alg1} has $2^{n(n+1)/2}$ states. Consequently, we have the following super-exponential worst-case time complexity of the algorithms.
 
	\begin{thm}
		The worst-case time complexity of Algorithms~\ref{alg-trellis} and~\ref{alg1} for partially ordered automata is $O^*{(2^{n(n+1)/2})}$, and this result is tight. \qed
	\end{thm}

	Similarly, we could discuss deterministic partially ordered automata, the strings over which correspond to binary upper-triangular matrices with up to one nonzero element in each row. The semigroup of such matrices has $(n+1)!$ elements and $n(n-1)/2$ generators, resulting in the tight worst-case time complexity $O^*((n+1)!)$ of Algorithms~\ref{alg-trellis} and~\ref{alg1} for deterministic partially ordered automata. Note that deterministic partially ordered automata are weaker than their nondeterministic counterpart, see \cite{mfcs16:mktmmt_full}.

\section{Experimental Comparison -- Benchmarks}\label{secExp}
	We implemented Algorithm~\ref{alg-trellis} in the \emph{trellis} tool and Algorithm~\ref{alg1} in the \emph{faudes} tool. We call the latter tool \emph{faudes} because it uses the core of the C++ library \texttt{libFAUDES}.\footnote{\url{https://fgdes.tf.fau.de/faudes/index.html}}
 	Algorithm~\ref{alg3} is a schema implemented as a transducer taking an IFO instance and creating a language-inclusion instance that is subsequently verified by an external tool. Included tools are \emph{vata} of \cite{AbdullaCHMV10}, \emph{limi} of \cite{CernyCHRRST17}, and \emph{mata} of \cite{mata}.\footnote{For technical problems and memory consumptions, we excluded \emph{hkc} of \cite{BonchiP13} and \emph{rabit} of \cite{ClementeM19}.} 

 	We ran experiments on an Ubuntu 22.04.3 LTS machine with 32 Intel(R) Xeon(R) CPU E5-2660 v2@2.20GHz processors and 246 GB memory, executing 30 experiments in parallel, using the \texttt{parallel} tool of \cite{tange_2024_10558745}. Each tool was given a timeout of five minutes. The data and tools are available at \url{https://apollo.inf.upol.cz:81/masopust/ifo-benchmarks}, where we further present results of experiments on random data.

	\begin{table}
	\ra{1.1}
	\centering
	\begin{tabular}{lrrrr}
		\toprule
			& \multicolumn{2}{c}{Negative instances}
			& \multicolumn{2}{c}{Positive instances}\\
			\cmidrule(lr){2-3}
			\cmidrule(lr){4-5}
			Tool & 1 min. & 5 min. & 1 min. & 5 min.\\
		\midrule
		vata		& 769  & 583  & 671 & 518 \\
		trellis		 & 120  & 64    & 867 & 720 \\
		faudes	 & 485   & 71    & 402 & 58 \\
		mata	  &  467 & 55     & 397 & 49 \\
		limi		 & 47    &  1      & 56   &  4\\
		\bottomrule
	\end{tabular}
	\medskip
	\caption{The numbers of instances not solved within a given time.}\label{failsTab}
	\end{table}

	Table~\ref{failsTab} summarizes the number of instances the tool did not solve within the given timeout. The best performance is obtained by \emph{limi}, which solved almost all instances within five minutes. It is worth noticing that \emph{trellis} performs very well for negative instances, but very badly for positive instances. This observation is in accordance with theoretical results where, for negative instances, the computation of \emph{trellis} stops as soon as a counterexample is found, whereas for positive instances, each element of the semigroup has to be verified. 
	
	\begin{table}
		\ra{1.1}
		\centering
		\begin{tabular}{@{}lrr@{}}
			\toprule
			Tool & \multicolumn{1}{c}{Negative instances}
			& \multicolumn{1}{c}{Positive instances}\\
			\midrule
			vata 
			& 51.233 / 298.907 / 0.004
			& 58.317 / 289.707 / 0.005\\
			trellis 
			& 14.985 / 298.882 / 0.003
			& 91.045 / 294.742 / 0.004\\
			faudes 
			& 102.447 / 295.533 / 0.006
			& 82.658 / 297.417 / 0.006\\
			mata 	
			& 81.044 / 299.024 / 0.006
			& 60.948 / 291.438 / 0.005\\
			limi 		
			& 20.555 / 135.465 / 0.188
			& 18.652 / 291.571 / 0.205\\
			\bottomrule
		\end{tabular}
		\smallskip
		\caption{Average/maximum/minimum computation times in seconds.}\label{averageTimeTab}
	\end{table}
	
	The average, maximum, and minimum computation times of the tools are summarized in Table~\ref{averageTimeTab}. On average, \emph{limi} performs the best on both negative and positive instances, while \emph{trellis} performs the best for negative instances. The reason why the performance of \emph{trellis} is bad for positive instances was discussed above. Its success for negative instances, on the other hand, may come from the possibility of quickly arriving at a counterexample if IFO is not satisfied. However, the question why {\em limi} performs so well on both types of instances is an open problem that may suggest the existence of a theoretically faster algorithm for IFO verification. 
	
	\begin{figure}
		\centering
  		\subfloat[Negative instances.]{\includegraphics[scale=.66]{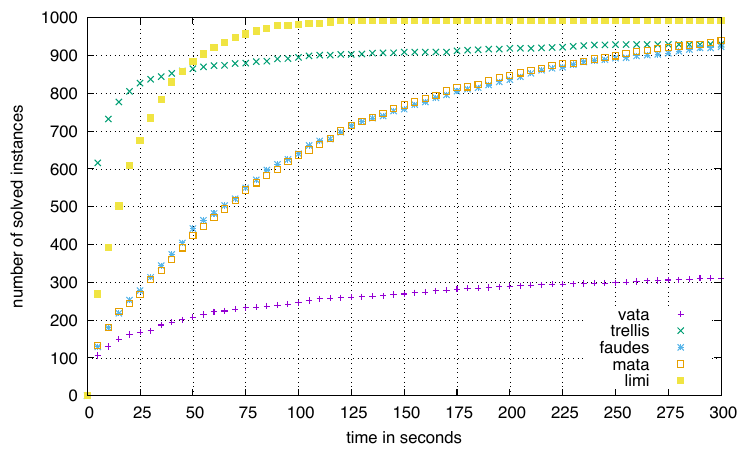}}\quad
  		\subfloat[Positive instances.]{\includegraphics[scale=.66]{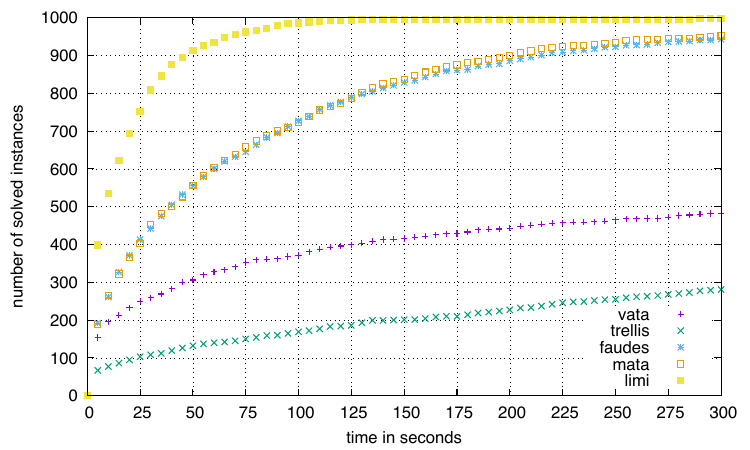}}
		\caption{Instances computed in five minutes. The $x$-axis displays the time (0--300 seconds), and the $y$-axis displays the number of solved instances in that time (0--994 for negative and 0--1001 for positive instances).}
  		\label{fig3}
	\end{figure}

	Figure~\ref{fig3} shows the number of instances solved by the tools in five minutes. It confirms that {\em limi} is the best tool to verify IFO, perhaps run with {\em trellis} in parallel to quickly catch negative instances. It is a challenging question whether the combination of the algorithms behind these tools may result in a better algorithm or help us answer the open problem of the complexity of the IFO-verification problem.

	Figure~\ref{fig5} shows the time to solve instances of particular sizes. Except for the case of \emph{trellis} for negative instances, the plots in a sense confirm the theoretical (super)exponential worst-case time complexity. Indeed, the reader can see that the plots resemble the positive part of the graph of an exponential function. The growth of the function is fastest for \emph{vata} and for positive instances of \emph{trellis}. On the other hand, for \emph{limi}, the growth is very slow, which makes the tool very attractive for the IFO verification.
	 
	\begin{figure}
		\begin{subfigure}{\columnwidth}
			\centering
			\includegraphics[scale=.32]{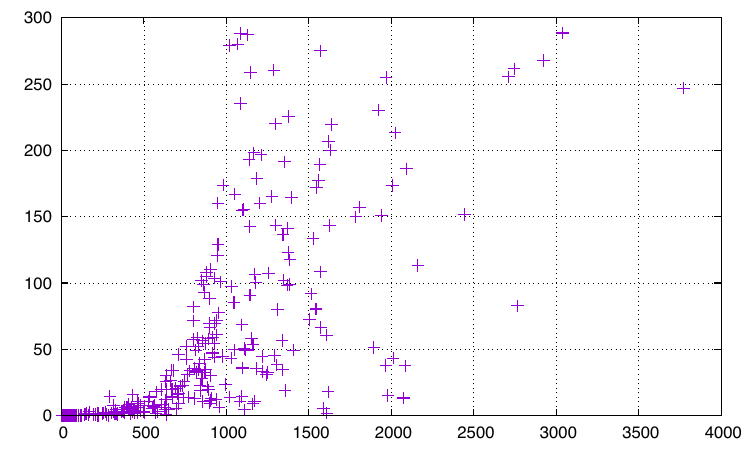}
			\includegraphics[scale=.32]{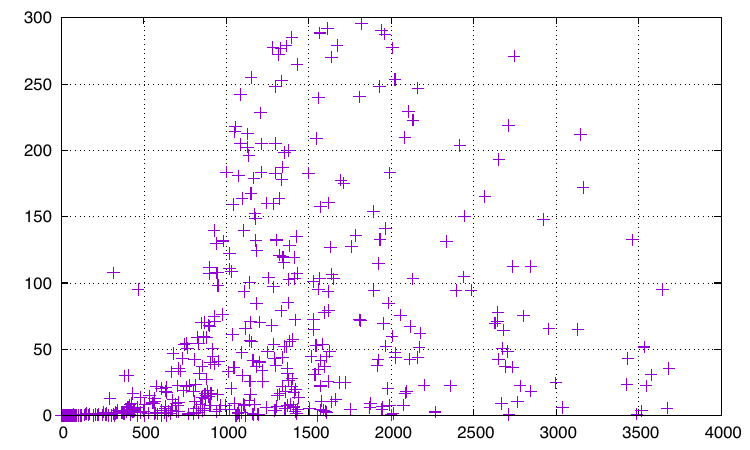}
			\caption{vata}
		\end{subfigure}
		\begin{subfigure}{\columnwidth}
			\centering
			\includegraphics[scale=.32]{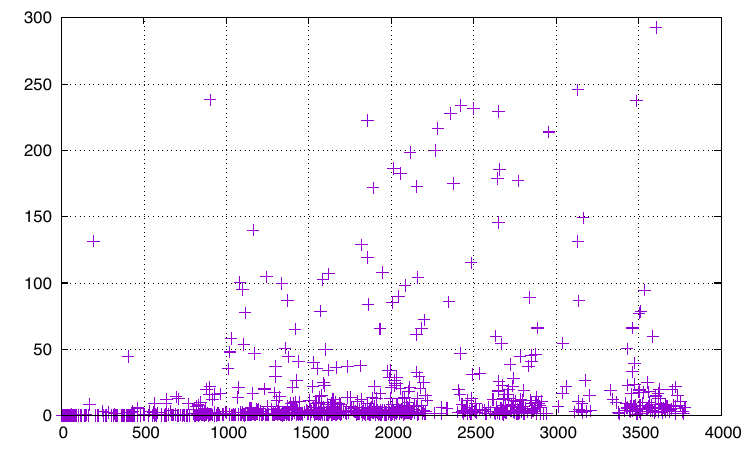}
			\includegraphics[scale=.32]{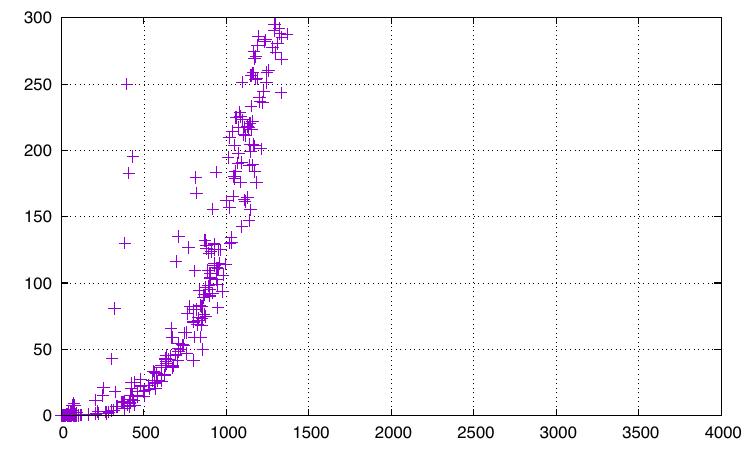}
			\caption{trellis}
		\end{subfigure}
		\begin{subfigure}{\columnwidth}
			\centering
			\includegraphics[scale=.32]{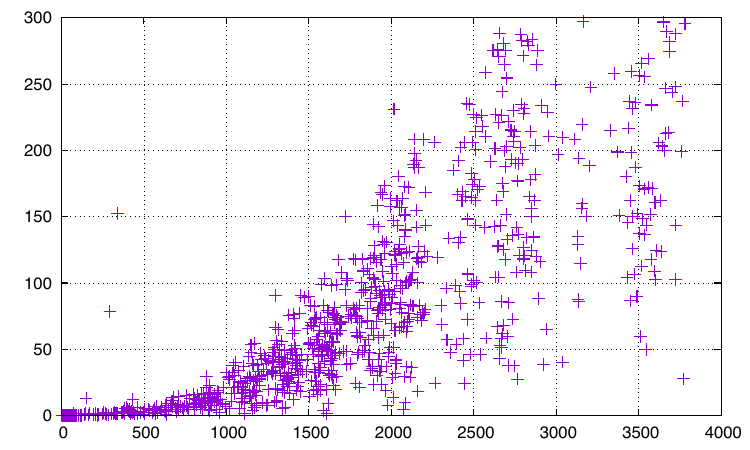}
			\includegraphics[scale=.32]{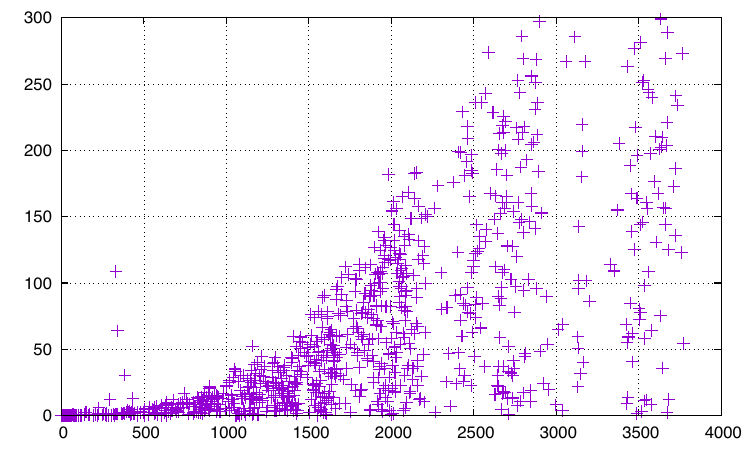}
			\caption{faudes}
		\end{subfigure}
		\begin{subfigure}{\columnwidth}
			\centering
			\includegraphics[scale=.32]{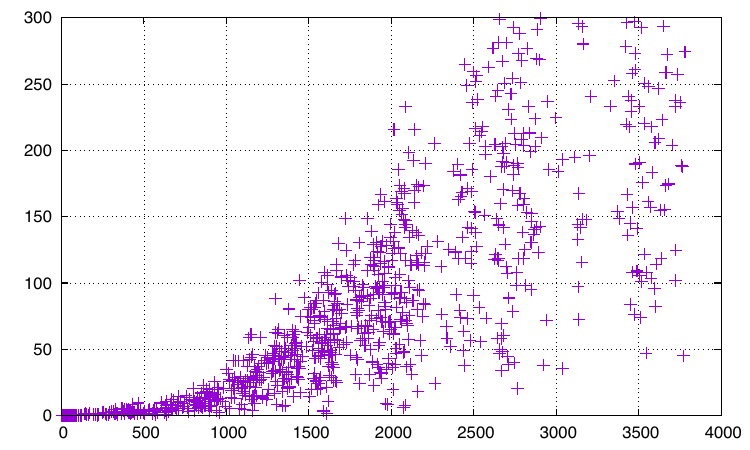}
			\includegraphics[scale=.32]{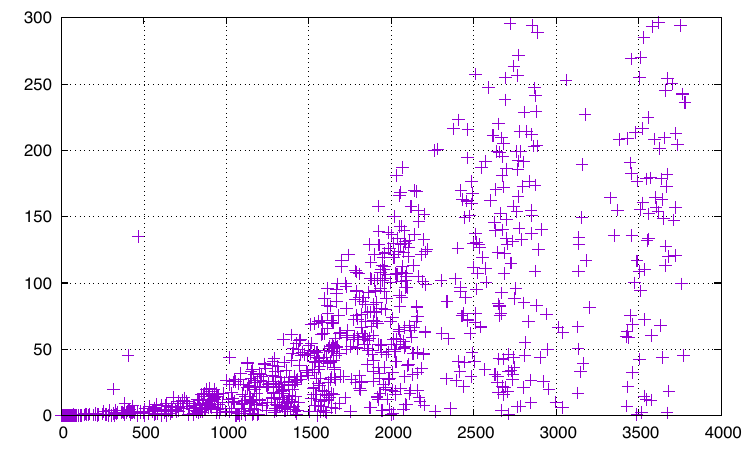}
			\caption{mata}
		\end{subfigure}
		\begin{subfigure}{\columnwidth}
			\centering
			\includegraphics[scale=.32]{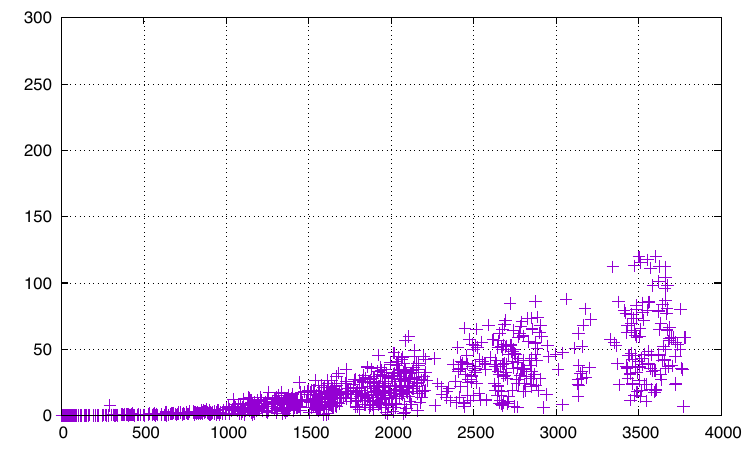}
			\includegraphics[scale=.32]{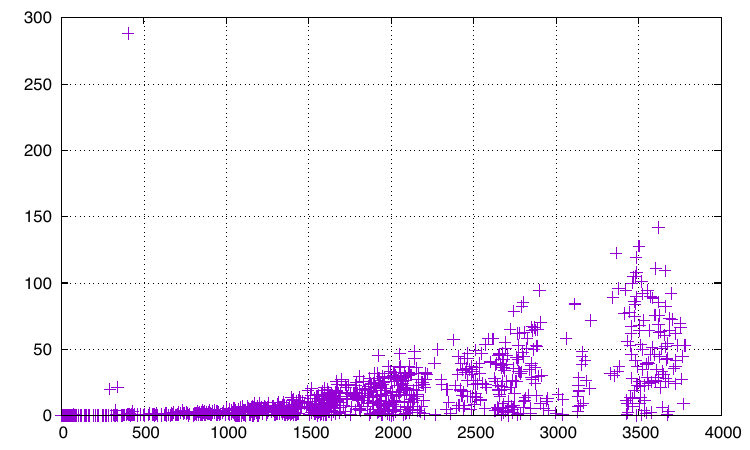}
			\caption{limi}
		\end{subfigure}
		\caption{Time to solve negative (left) and positive (right) instances. The $x$-axis displays the size of the instance (0--4000 states), and the $y$-axis displays the time to solve the instance (0--300 s).}
		\label{fig5}
	\end{figure}

	We would like to point out that the results on random data, presented at the git repository, show the same results. There, we generated uniformly at random 300 IFO instances for each of $250, 500, 750, \dots, 6000$ states, which resulted in 7200 random IFO instances. 
	
	Quite an efficient behavior of the advanced inclusion-based tools, in particular of the \emph{limi} tool, compared with the \emph{trellis} tool, is of particular significance, namely if the reader realizes that the inputs for the inclusion-based tools are of size quadratically larger than the inputs for {\em trellis}. In many cases, they have millions of states and even more transitions.

	Finally, we would like to mention that the \emph{trellis} and \emph{faudes} tools implementing the textbook algorithms (Algorithm~\ref{alg-trellis} and Algorithm~\ref{alg1}) are on purpose without any optimizations. Our results thus in no way say anything about (in)efficiency of the \texttt{libFAUDES} library. 

\begin{ack}
  We gratefully acknowledge suggestions and comments of the anonymous referees.
\end{ack}

\bibliographystyle{elsarticle-harv}
\bibliography{mybib}

\end{document}